\documentclass[12pt, a4paper]{article}
\usepackage[margin=3cm]{geometry}
\usepackage{amscd, amsmath, amssymb, amsthm, appendix, authblk, enumerate, enumitem, float, mathtools,  subcaption, tabu, tikz}
\usetikzlibrary{matrix,arrows,positioning,automata,shapes,shapes.multipart} 

\newtheorem{theorem}{Theorem}[section]
\newtheorem{proposition}[theorem]{Proposition}
\newtheorem{lemma}[theorem]{Lemma}

\newtheorem*{theorem*}{theorem}

\theoremstyle{definition}
\newtheorem{definition}[theorem]{Definition}
\newtheorem*{definition*}{Definition}

\newtheorem{question}[theorem]{Question}
\newtheorem*{question*}{Question}

\newtheorem{remark}[theorem]{Remark}
\newtheorem{example}[theorem]{Example}

\begin{document}

\title{Tree-Based Unrooted Nonbinary Phylogenetic Networks}
\author{Michael Hendriksen}
\affil{\textit{Centre for Research in Mathematics, Western Sydney University, NSW, Australia}}
\date{}

\maketitle

\begin{abstract}
Phylogenetic networks are a generalisation of phylogenetic trees that allow for more complex evolutionary histories that include hybridisation-like processes. It is of considerable interest whether a network can be considered `tree-like' or not, which lead to the introduction of \textit{tree-based} networks in the rooted, binary context. Tree-based networks are those networks which can be constructed by adding additional edges into a given phylogenetic tree, called the \textit{base tree}. Previous extensions have considered extending to the binary, unrooted case and the nonbinary, rooted case. We extend tree-based networks to the context of unrooted, nonbinary networks in three ways, depending on the types of additional edges that are permitted. A phylogenetic network in which every embedded tree is a base tree is termed a \textit{fully tree-based} network. We also extend this concept to unrooted, nonbinary phylogenetic networks and classify the resulting networks. In addition, we derive some results on the colourability of tree-based networks, which can be useful to determine whether a network is tree-based.
\end{abstract}

\section{Introduction}
There is some discussion within the biological community about whether certain evolutionary histories are `tree-like' with some reticulation, or whether their history is not tree-like at all  \cite{Corel2016,Dagan2006,Doolittle2007,Martin2011}. This question has recently lead to the introduction of the concept of \textit{tree-based} networks by Francis and Steel in \cite{Francis2015}, which, roughly speaking, are phylogenetic trees with additional arcs placed between edges of the tree. In particular, the concept was introduced and applied to binary, rooted phylogenetic networks.

More recently, Francis, Huber and Moulton in \cite{Francis2017} extended the concept of tree-basedness to binary unrooted networks, and Jetten and van Iersel in \cite{Jetten2016} extended the concept to nonbinary rooted networks. In particular, in the nonbinary setting, it is possible to place an arc between an edge and a non-leaf vertex of the base tree (producing a vertex with at least degree $4$), which was not possible in the binary case. Jetten and van Iersel therefore distinguished the concepts of \textit{strictly tree-based} networks and \textit{tree-based networks} in the nonbinary setting, where in the former case edge-to-vertex edges are not permitted. 

A possible issue with considering whether a given evolutionary history is tree-like is that under the definition of a tree-based network it is possible for multiple non-isomorphic trees to be a base tree for a given network. In this circumstance, while a network may have a reasonable claim to be tree-like, a claim that it is like a \textit{particular} tree is much harder to define.

The issue is particularly magnified by the possibility for every single tree embedded in a network to be a base tree. Semple in \cite{Semple2016} showed that for binary rooted networks, the class of tree-based networks for which every embedded tree is a base tree (later termed \textit{fully tree-based} networks in \cite{Francis2017}) coincides with the familiar class of tree-child networks. Later, Francis, Huber and Moulton in \cite{Francis2017} showed that in the binary unrooted setting, a network is fully tree-based if and only if it is a level-$1$ network.

This article extends the study of tree-based networks to the \textit{non-binary, unrooted} setting. Preliminary phylogenetic and graph-theoretic terminology is introduced in Section $2$. In Section $3$, we define strictly tree-based and tree-based networks in the unrooted nonbinary setting, and introduce a third analogue, the \textit{loosely tree-based} network. These are then characterised in terms of spanning trees of the network, generalising the characterisation by Francis, Huber and Moulton in  \cite{Francis2017}.  In Section $4$ we then characterise the nonbinary unrooted analogues of fully tree-based networks and provide some constructions for these networks of arbitrary level where this is possible. In Section $5$ we end with some results on colourability of tree-based networks, which can assist in identifying networks that are not tree-based or not strictly tree-based.

\section{Phylogenetic Networks}

In this article we examine phylogenetic networks that are unrooted and nonbinary, so our definition reflects this, in contrast to \cite{Francis2015} and \cite{Jetten2016}.

\begin{definition}
Let $X$ be a finite, non-empty set. A \textit{phylogenetic network} is a connected graph $(V,E)$ with $X \subseteq V$ and no degree $2$ vertices, such that the set of degree $1$ vertices (referred to as \textit{leaves}) is precisely $X$. A \textit{phylogenetic tree} is a phylogenetic network that is also a tree.
\end{definition}

Throughout this paper, all trees and networks are unrooted, nonbinary phylogenetic trees and networks unless otherwise specified. We note that by nonbinary we do not mean that there must be vertices of degree $4$ or more, just that they are permitted. If $N$ does contain a vertex of degree $4$ or more, then we say that $N$ is \textit{strictly nonbinary}.

To be able to define tree-based networks in the next section, it is necessary to define an operation used on phylogenetic networks in their construction.

\begin{definition}
Let $N$ be a network with some edge $e=(v,w)$. The operation in which we delete $e$, add a new vertex $z$, then add the edges $(v,z)$ and $(z,w)$ is termed \textit{subdividing} $e$, and the new vertex $z$ is referred to as an \textit{attachment point}. The inverse operation of subdivision is termed \textit{suppressing} a vertex, in which a degree $2$ vertex $z$ with edges $(v,z)$ and $(z,w)$ is deleted and then the edge $(v,w)$ is added.
\end{definition}

 Recall the following standard definitions from \cite{Gambette2012}.

\begin{definition}
Let $N$ be a network on edge set $X$. A \textit{cut-edge} is an edge such that $N-e$ is a disconnected graph. A \textit{pendant} edge is a cut-edge where one of the connected components of $N-e$ is trivial. We refer to $N$ as \textit{simple} if all cut-edges of $N$ are pendant. A \textit{blob} is a maximal connected subgraph of $N$ with no cut-edges that is not a vertex.
\end{definition}

Given a network $N$ and a blob $B$ in $N$, we define a simple network $B_N$ by taking the union of $B$ and all cut-edges in $N$ incident with $B$, where the leaf-set of $B_N$ is just the set of end vertices of these cut-edges that are not already a vertex in $B$.

In order to classify fully tree-based networks and loosefully tree-based networks in Section 4, we require one final definition.

\begin{definition}
Let $N$ be a network and $v$ a vertex of $N$. We say that $v$ is a \textit{cut-vertex} if deletion of $v$ and all edges incident to $v$ from $N$ forms a disconnected graph. The connected components of this disconnected graph are referred to as the  \textit{cut-components} of $v$ in $N$.
\end{definition}

Note that all vertices with a pendant edge are cut-vertices, but not all cut-vertices have pendant edges. For an example, see Figure \ref{KFully} (ii), where the central black vertex is a cut-vertex but has no pendant edges.

\section{Nonbinary Unrooted Tree-Based Networks}

Tree-based networks were initially introduced by Francis and Steel in  \cite{Francis2015}, in which they were constructed by subdividing some number of edges of a binary, rooted phylogenetic tree, then adding additional edges between pairs of attachment points so that only one new edge is added to each attachment point.

In the nonbinary case, we have more freedom. In this case, we can additionally add edges between attachment points and the original vertices of the tree, or even between two vertices in the original tree, as we no longer need to worry about the resulting network being binary. \cite{Jetten2016} consider this idea in the rooted, nonbinary setting. If the new edges are exclusively between attachment points with no more than one edge at each attachment point, they termed the network \textit{strictly tree-based}. If edges between attachment points and vertices of the original tree are allowed, they termed the network \textit{tree-based}. We will formally define these in the unrooted nonbinary case shortly. In this article, we will additionally consider the possibility in which we allow edges between two vertices in the original tree and more than one additional edge incident to an attachment point.

In \cite{Francis2017}, a binary, unrooted tree-based network is defined to be a network $N$ on $X$ which contains a spanning tree on the same leaf-set $X$. If we consider a binary network $N$ with some spanning tree $T$ and some edge $e \in N-T$, then $e$ can only be incident to degree $2$ vertices of $T$. If $e$ were incident to some degree $1$ vertex of $T$ then $N$ and $T$ would not have the same leaf set, and if $e$ were incident to some vertex of degree $3$ or more in $T$, then $N$ would be strictly nonbinary, contradicting the fact that $N$ is binary. Additionally, no pair of edges $e_1, e_2 \in N-T$ can be incident to the same vertex $v$ of $T$, as this would force $v$ to have degree $4$ or more. Observe that, with the exception of the case where $e$ is incident to a degree $1$ vertex, the limitation was due to $N$ being binary. 

Therefore, if $N$ is a nonbinary, unrooted network on $X$, then given some spanning tree $T$ of $N$ on $X$, some edge $e \in N-T$ may be between a pair of degree $2$ vertices of $T$, between a degree $2$ vertex and a degree $3$ or more vertex, or between a pair of degree $3$ or more vertices. From the point of view of constructing $N$ from a base tree $T$, this respectively coincides with adding an edge between attachment points, between an attachment point and an original vertex of the tree, or between two original vertices of the tree. We will refer to networks that satisfy this nonbinary analogue of the spanning tree definition from \cite{Francis2017} by the term \textit{loosely tree-based} networks.

As spanning trees are well-known and well-understood, we formally define loosely tree-based networks, as well as the nonbinary unrooted analogues of tree-based and strictly tree-based networks, in terms of spanning trees. We will then show that these definitions are equivalent to the attachment point definitions.

\begin{definition}
\label{SpanDef}
Let $N$ be a network on leaf-set $X$. Then $N$ is referred to as

\begin{enumerate}
\item \textit{loosely tree-based} on $X$ if there exists a spanning tree in $N$ whose leaf-set is equal to $X$,
\item \textit{tree-based} on $X$ if there exists a spanning tree $T$ of $N$ such that $T$ contains all edges between two vertices of degree $4$ or more, and all degree $2$ vertices of $T$ were degree $3$ in $N$, and
\item \textit{strictly tree-based} on $X$ if there exists a spanning tree $T$ in $N$ whose leaf-set is equal to $X$ and $T$ contains every edge incident to the vertices of degree at least $4$.
\end{enumerate}
\end{definition}

It follows that the class of strictly tree-based networks is contained in the class of tree-based networks, which are in turn contained in the class of loosely tree-based networks. The distinction between the three types of tree-basedness is shown by the networks in Figure \ref{TBExamples}.

\begin{figure}[H]
\centering
\begin{tikzpicture}
\draw (0,1) --(1,2)--(2,1)--(3,2)--(4,1);
\draw (0,3) --(1,2)--(3,2)--(4,3);
\draw (3,0)--(2,1)--(1,0);

\draw[fill] (1,2) circle [radius=1.5pt];
\draw[fill] (2,1) circle [radius=1.5pt];
\draw[fill] (3,2) circle [radius=1.5pt];
\draw[fill] (0,1) circle [radius=1.5pt];
\draw[fill] (4,1) circle [radius=1.5pt];
\draw[fill] (0,3) circle [radius=1.5pt];
\draw[fill] (4,3) circle [radius=1.5pt];
\draw[fill] (3,0) circle [radius=1.5pt];
\draw[fill] (1,0) circle [radius=1.5pt];

\node[above=-0.1cm] at (2,2) {$\mathstrut e$};
\node[below left=-0.1cm] at (1.5,1.5) {$\mathstrut e$};
\node[below right=-0.1cm] at (2.5,1.5) {$\mathstrut e$};
\node at (2,-1) {$\mathstrut (i)$};
\end{tikzpicture}
\hspace{1cm}
\begin{tikzpicture}
\draw (0,1) --(1,2)--(2,1)--(3,2)--(4,1);
\draw (0,3) --(1,2)--(3,2)--(4,3);
\draw (2,1)--(2,0);

\draw[fill] (1,2) circle [radius=1.5pt];
\draw[fill] (2,1) circle [radius=1.5pt];
\draw[fill] (3,2) circle [radius=1.5pt];
\draw[fill] (0,1) circle [radius=1.5pt];
\draw[fill] (4,1) circle [radius=1.5pt];
\draw[fill] (0,3) circle [radius=1.5pt];
\draw[fill] (4,3) circle [radius=1.5pt];
\draw[fill] (2,0) circle [radius=1.5pt];

\node[below left=-0.1cm] at (1.5,1.5) {$\mathstrut e$};
\node[below right=-0.1cm] at (2.5,1.5) {$\mathstrut e$};
\node at (2,-1) {$\mathstrut (ii)$};
\end{tikzpicture}
\hspace{1cm}
\begin{tikzpicture}
\draw (0,3) --(1,2)--(2,1)--(3,2)--(4,1);
\draw (1,2)--(3,2)--(4,3);
\draw (2,1)--(2,0);

\draw[fill] (1,2) circle [radius=1.5pt];
\draw[fill] (2,1) circle [radius=1.5pt];
\draw[fill] (3,2) circle [radius=1.5pt];
\draw[fill] (0,3) circle [radius=1.5pt];
\draw[fill] (4,1) circle [radius=1.5pt];
\draw[fill] (4,3) circle [radius=1.5pt];
\draw[fill] (2,0) circle [radius=1.5pt];

\node[below left=-0.1cm] at (1.5,1.5) {$\mathstrut e$};
\node at (2,-1) {$\mathstrut (iii)$};
\end{tikzpicture}
\caption{(i) A loosely tree-based network that is not tree-based; (ii) A network that is tree-based but not strictly tree-based; (iii) A strictly tree-based network. The edges labelled $e$ are possible edges that were added to a base tree to construct the network.}
\label{TBExamples}
\end{figure}
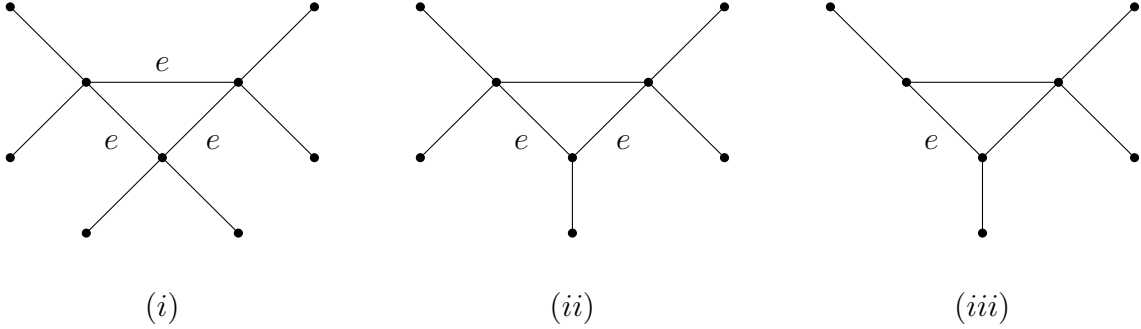

We note that these definitions provide immediate access to some insights that may be tougher to see using a construction-based definition. For example, we can see that if a network contains a cycle consisting of vertices of degree $4$ it cannot be tree-based, and if it contains a cycle that does not have two adjacent degree $3$ vertices it cannot be strictly tree-based.

We will now show in Theorems \ref{TreeBaseSpan} and \ref{StrictBaseSpan} that our spanning tree definitions are equivalent to the familiar tree-based definitions from previous papers.

\begin{theorem}
\label{TreeBaseSpan}
Let $N$ be a network on a leaf-set $X$. Then the following are equivalent:

\begin{enumerate}
\item $N$ is tree-based.
\item $N$ can be obtained by taking a tree $T$, subdividing edges of $T$ to form attachment points and adding edges either between attachment points or between an attachment point and an original vertex of $T$ (so that each attachment point now has degree $3$). 
\end{enumerate}
\end{theorem}

\begin{proof}
Suppose $N$ was obtained by taking a tree $T$ and performing the above procedure. Let $T^+$ be the tree $T$ with the required attachment points added. Then $T^+$ is necessarily a spanning tree of $N$ on $X$, as no step in the procedure adds vertices or adds an edge to a leaf. As we add precisely one edge to each degree $2$ vertex (that is, attachment point) of $T^+$, all vertices of $T^+$ of degree $2$ are degree $3$ in $N$. Furthermore, no edges are added between a pair of vertices of degree $3$ or more in $N$, so $T^+$ contains all edges of $N$ that lie between two vertices of degree $4$ or more. Hence $N$ is tree-based.

Now suppose $N$ is tree-based, so there exists a spanning tree $T$ of $N$ with leaf-set $X$ such that all degree $2$ vertices in $T$ were degree $3$ vertices in $N$, and $T$ contains all edges between two vertices of degree $4$ or more. Denote the tree obtained by suppressing any degree $2$ vertices of $T$ by $T^-$. Then we can subdivide edges of $T^-$ to make $T$, and add edges between the attachment points and either other attachment points or original vertices to make $N$ (keeping each attachment point to degree $3$). As all edges deleted from $N$ to make $T$ are incident to at least one degree $3$ vertex, this means all required additional edges are incident to an attachment point. As all degree $2$ vertices in $T$ were degree $3$ vertices in $N$ it is not possible for two edges in $N-T$ to be incident to the same attachment point. It follows that $T^-$ is the tree required by the theorem.
\end{proof}

We now consider strictly tree-based networks.

\begin{theorem}
\label{StrictBaseSpan}
Let $N$ be a  network on leaf-set $X$. Then the following are equivalent:

\begin{enumerate}
\item $N$ is strictly tree-based.
\item $N$ can be obtained by taking a tree $T$ , subdividing edges of $T$ to form attachment points, and adding edges between those attachment points (so that each attachment point has degree $3$).
\end{enumerate}
\end{theorem}

\begin{proof}
Suppose $N$ is strictly tree-based, so contains a spanning tree $T$ that includes all edges incident to vertices of degree $4$ or more. Then all edges in $N-T$ are between vertices of degree $3$ in $N$, as all edges incident to vertices of degree $1$ or degree $4$ are in $T$, and there are no vertices of degree $2$. 

Let $e$ be an edge in $N-T$. There cannot be two edges in $N-T$ incident to the same degree $3$ vertex, as otherwise either $T$ contains a leaf that $N$ does not or $T$ does not span every vertex. Hence the deletion of $e$ from $N$ causes there to be two degree $2$ vertices, which may be suppressed to obtain an edge. It follows that $N$ is obtained by taking a tree, subdividing edges to form attachment points and adding edges between those attachments points (so that each attachment point has degree $3$).

Now suppose $N$ is obtained by taking some tree $T$, then subdividing edges of $T$ and connecting those vertices obtained by subdivision. Let the subdivided subtree of $N$ corresponding to $T$ be denoted $T^+$. Then clearly $T^+$ is a spanning tree, as it contains every vertex in $N$. We then observe that every vertex of degree $4$ or more is contained within the spanning tree, since the additional arcs are only incident to vertices of degree $3$.
\end{proof}

We now provide an example of a strictly nonbinary network that is not tree-based.

\begin{example}
Figure \ref{NotLoose} provides an example of a strictly nonbinary level-$5$ network that is not loosely tree-based. Observe that there is no Hamiltonian path between the two leaves. As such there can be no spanning tree on the same leaf set. This example also demonstrates the distinction between containing a spanning tree \textit{on the same leaf set} and merely having a spanning tree, as this example contains several spanning trees.

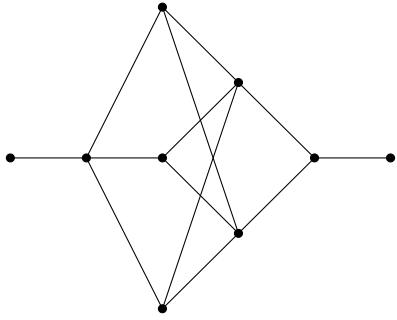
\begin{figure}[H]
\centering
\begin{tikzpicture}
\draw (0,2) --(2,2)--(3,3)--(4,2)--(5,2);
\draw (1,2) --(2,4) --(3,3);
\draw (2,4) --(3,1)--(4,2);
\draw (2,2) --(3,1);
\draw (1,2) --(2,0);
\draw (2,0) --(3,3);
\draw (2,0) --(3,1);

\draw[fill] (0,2) circle [radius=1.5pt];
\draw[fill] (1,2) circle [radius=1.5pt];
\draw[fill] (2,2) circle [radius=1.5pt];
\draw[fill] (2,4) circle [radius=1.5pt];
\draw[fill] (2,0) circle [radius=1.5pt];
\draw[fill] (3,3) circle [radius=1.5pt];
\draw[fill] (3,1) circle [radius=1.5pt];
\draw[fill] (4,2) circle [radius=1.5pt];
\draw[fill] (5,2) circle [radius=1.5pt];
\end{tikzpicture}
\caption{Example of an unrooted strictly nonbinary network that is not loosely tree-based.}
\label{NotLoose}
\end{figure} 
\end{example}

Francis, Huber and Moulton prove a decomposition theorem for unrooted tree-based binary networks in \cite{Francis2017}. A similar one exists for nonbinary tree-based networks. Recall that a blob $B$ is a maximal connected subgraph of $N$ with no cut-edges that is not a vertex, and that the simple network $B_N$ is constructed by taking the union of $B$ and all cut-edges in $N$ incident with $B$, where the leaf-set of $B_N$ is just the set of end vertices of these cut-edges that are not already a vertex in $B$.

\begin{proposition}
Suppose $N$ is a network. Then $N$ is loosely, normally or strictly tree-based if and only if $B_N$ is (respectively) loosely, normally or strictly tree-based for every blob $B$ in $N$.
\end{proposition}

\begin{proof}
Similar proof to Proposition $1$ in \cite{Francis2017}. 
\end{proof}

Again as in \cite{Francis2017} we can immediately classify networks on a single leaf.

\begin{remark}
Suppose $N$ is a network on $\{ x \}$. Then $N$ is tree-based if and only if $N= (\{ x \}, 0)$.
\end{remark}

\section{Fully Tree-Based Networks}

A \textit{fully tree-based} network $N$ on leaf-set $X$ is a network where every embedded tree with leaf-set $X$ is a base tree, with the original concept coming from \cite{Semple2016} and the terminology from \cite{Francis2017}. Of course, in the nonbinary setting we must be clear about what sort of base tree we are referring to - strict, normal or loose. 

\begin{definition}
Let $N$ be a network on leaf-set $X$. Then $N$ is \textit{strictfully, fully,} or \textit{loosefully} tree-based if every embedded tree with leaf-set $X$ is a base tree in the (respectively) strict, usual or loose sense.
\end{definition}

 In the binary case, a network is fully tree-based if and only if it is a level-$1$ network \cite{Francis2017}. Correspondingly, we will show that in the nonbinary unrooted case, a simple network $N$ is \textit{strictfully} tree-based if and only if it is a binary level-$1$ network or a star tree. In general this means that a network $N$ is strictfully tree-based if and only if for all blobs $B$ of $N$, $B_N$ is a binary level-$1$ network.

For example, let $T$ be a strictly nonbinary tree such that there is a trio of vertices $v_1, v_2, v_3$ of degree $1$ or $3$ with edges $e_1=(v_1,v_2)$ and $e_2 = (v_2,v_3)$. Then by adding an attachment points each to $e_1$ and $e_2$ and connecting them to form a network $N$, we obtain a biconnected component $B$ so that $B_N$ is a binary level $1$ network (even though $N$ is not binary). It follows that $N$ is strictfully tree-based as the only blob in $N$ is binary and level-$1$. We will now find characterisations of simple strictfully, fully and loosefully tree-based networks, from which similar general results can be drawn.

We make the following fairly trivial but extremely useful remark

\begin{remark}
\label{SubEmbed}
Let $N$ be a network on leaf-set $X$, possibly with degree $2$ vertices. Suppose $S$ is a connected subgraph of $N$ with leaf-set $X$. Then $S$ contains an embedded tree with leaf-set $X$, contained within the network obtained by the union of the shortest paths between leaves.
\end{remark}

We will now prove our statement regarding strictfully tree-based networks, which is the direct analogue of the binary result on fully tree-based networks from \cite{Francis2017}.

\begin{theorem}
Let $N$ be a simple network. Then $N$ is strictfully tree-based if and only if $N$ is a level-$1$ binary network or a star tree.
\end{theorem}

\begin{proof}
Star trees are obviously strictfully tree-based, and any tree that is not a star tree is not simple. We therefore assume that $N$ is a strictfully tree-based level-$k$ network for $k > 0$. Suppose $N$ is strictly nonbinary, so there exists some non-leaf vertex $v$ of degree $4$ or more. Then $v$ have at least two non-pendant edges as $N$ is simple. Label one of them $e$, and observe that $N-e$ is a connected graph with leaf-set $X$. By remark \ref{SubEmbed}, $N-e$ contains a subtree $T$ on leaf-set $X$. As $N$ is strictfully tree-based, $T$ must be a spanning tree. But $T$ is then a spanning tree that does not include all edges incident to vertices of degree $4$ or more, so $N$ is not strictfully tree-based.

Therefore all simple strictfully tree-based networks are binary networks. In the binary case, the definition of strictfully tree-based coincides with the definition of tree-based, and Theorem 5 from \cite{Francis2017} states that binary phylogenetic networks are tree-based if and only if they are level-$1$. The theorem follows.
\end{proof}

Recall that a vertex $v$ in a network $N$ is a cut-vertex if deletion of $v$ and all edges incident to $v$ from $N$ forms a disconnected graph.

\begin{lemma}
\label{CutInSub}
Let $N$ be a simple loosefully tree-based network on $X$ with a cut-vertex $v$. Then any embedded tree in $N$ on $X$ contains $v$.
\end{lemma}

\begin{proof}
We claim that every cut-component of $v$ in $N$ contains at least one leaf. If this is true, then suppose $a$ and $b$ are leaves from separate cut-components. Then any path between $a$ and $b$ must pass through $v$, so any tree on $X$ (which must include $a$ and $b$) must include $v$ as well. It remains to be shown that all cut-components contain at least one leaf.

Suppose some cut-component of $v$ in $N$ did not contain a leaf, but contains some vertex $w$. Then if we take any pair of leaves $x$ and $y$, a path between them cannot contain $w$ as the path would need to contain $v$ twice. As the union of paths $p_{xy}$ for all leaves $x,y \in X$ induces a network that contains a tree on the leaf-set $X$ that does not contain $w$, $N$ is not loosefully tree-based. Thus every cut-component formed by deleting $v$ contains a leaf and the result is proven.
\end{proof}

The requirement that all cut-vertices must be in any embedded tree proves useful for both of our remaining classifications. Unfortunately, fully tree-based networks and loosefully tree-based networks are not as uncomplicated as strictfully tree-based networks.

\begin{theorem}
\label{LoosefullyTreeBased}
Let $N$ be a simple network on taxa $X$. Then $N$ is loosefully tree-based if and only if every non-leaf vertex is a cut-vertex.
\end{theorem}

\begin{proof}
The level-$0$ case is trivial, so we assume $N$ is level-$k$ for $k \ge 1$.

Suppose all nonleaf vertices of $N$ are cut-vertices. Then any subtree of $N$ on $X$ must contain all vertices of $N$ by Lemma \ref{CutInSub}. It follows that any subtree of $N$ on $X$ must also be a spanning tree of $N$ on $X$.

Now suppose that $N$ is loosefully tree-based. Seeking a contradiction, let $v$ be a non-leaf vertex that is not a cut-vertex. Then $N-v$ is a connected network on leaf-set $X$ and hence there exists a subtree of $N-v$ on leaf-set $X$. This implies that $N$ contains a subtree $T$ on leaf-set $X$ that does not contain $v$, and hence $T$ is not a spanning tree. It follows that $N$ is not loosefully tree-based, a contradiction. Therefore all vertices of $N$ are cut-vertices.
\end{proof}

We note that this makes it fairly trivial to find loosefully tree-based networks of level-$k$ for any $k \ge 0$. It suffices to find a loosely tree-based level-$k$ network $N$, then add leaves to each vertex in $N$ that does not already have a pendant edge.

Finally, we classify fully tree-based networks.

\begin{theorem}
Let $N$ be a simple network on taxa $X$. Then $N$ is fully tree-based if and only if 

\begin{enumerate}
\item every non-leaf vertex in $N$ either is degree $3$ with a pendant edge or is a cut-vertex with at least $3$ cut-components, and
\item every vertex of degree $4$ or more in $N$ is only adjacent to vertices of degree $3$ or $1$.
\end{enumerate}
\end{theorem}

\begin{proof}
The level-$0$ case is trivially true, so we assume that $N$ is level-$k$ for $k \ge 1$.

Suppose the vertices obey the conditions outlined in the statement. Let $T$ be a subtree of $N$. We again can see that $T$ must be a spanning tree, as all non-leaf vertices of $N$ are cut-vertices. Suppose $v$ is a vertex in $N$. We note that $v$ must be adjacent in $T$ to at least one vertex from each of its cut components, or otherwise $T$ does not span $N$. Hence if $v$ has at least $3$ cut-components it must have degree at least $3$ in $T$. Otherwise $v$ has a pendant edge and one other cut-component. As $v$ can be adjacent to at most two vertices in its non-pendant cut-component, $v$ can have degree at most $3$. It follows that all vertices of degree $2$ in $T$ were degree $3$ in $N$. 

Finally, as all vertices of degree $4$ or more are only adjacent to vertices of degree $3$ or $1$, all spanning trees contain all edges of $N$ between two adjacent vertices of degree $4$ or more (as there are none). It follows that $N$ is fully tree-based.

Now suppose $N$ is fully tree-based. We will show that neither property $1$ nor $2$ from the statement of the theorem are possible. First suppose $2$ holds, that is, let $N$ contain a pair of adjacent vertices of degree $4$ or more and denote their shared edge $e$. We see that $N-e$ must be a connected subgraph (or $e$ would be a cut-edge), so by remark \ref{SubEmbed} $N-e$ contains a subtree on leaf-set $X$. Thus $N$ contains a subtree $T$ on leaf-set $X$ that does not include $e$, and since $N$ is fully tree-based, $T$ is a spanning tree. Hence there is a spanning tree of $N$ that does not use $e$ (an edge between two vertices of degree $4$ or more), so $N$ is not fully tree-based, a contradiction. We therefore assume all vertices of degree $4$ or more in $N$ are only adjacent to vertices of degree $3$ or $1$. 

We now note that being fully tree-based is a stronger condition than being loosefully tree-based, so by Theorem \ref{LoosefullyTreeBased} every vertex of $N$ must be a cut-vertex. 

Suppose that $v$ is a cut-vertex with precisely $2$ cut-components. We claim that $v$ must have a pendant edge, thus meeting the requirements in the statement of the theorem. Let the set of vertices adjacent to $v$ in one cut-component be $A_1=\{a_1,...,a_s\}$ and let the vertices adjacent to $v$ in the other be $A_2=\{b_1,...,b_t\}$. We can assume that either $A_1$ or $A_2$ contains more than one vertex, as $v$ cannot have degree $2$. Therefore suppose $|A_1| > 1$. As $A_1$ is in a single cut component, there must exist a path between $a_1$ and each of $a_2,...,a_s$ that does not contain $v$. Similarly, if $|A_2| > 1$,  there must exist a path between $b_1$ and each of $b_2,...,b_t$ that does not contain $v$.

Denote the edge between $v$ and $a_i$ by $e_i$, and the edge between $v$ and $b_i$ by $f_i$. Additionally, let $C=\{e_2,...,e_s, b_2,...,b_t\}$, interpreted as just $\{e_2,...,e_s \}$ if $|A_2|=1$. Then $N-C$ is a connected subgraph with leaf-set $X$, so $N-C$ contains a subtree on leaf-set $X$. This means $N$ contains a subtree $T$ on leaf-set $X$ that does not contain $e_2,...,e_s,b_2,...,b_t$, so $v$ is degree $2$ in $T$, which is in fact a spanning tree as $N$ is fully tree-based.

We note that if $|A_i| > 1$ for both $i=1,2$, or if $|A_i| > 2$ for either $i=1$ or $i=2$, $v$ has degree $4$ or more. In either case we obtain a spanning tree for which $v$ has degree $4$ or more in $N$ but $2$ in the spanning tree, so $N$ is not fully tree-based by definition. Hence $|A_1|=2$ and $|A_1|=1$, so if $v$ is a cut-vertex with $2$ cut-components, $v$ must be degree $3$ and have a pendant edge.

Hence $N$ meets the conditions outlined in the statement.
\end{proof}

It is still rather easy to construct a level-$k$ fully tree-based network for any $k \ge 0$. Figure \ref{TBExamples} (iii) provides an example of a fully tree-based level-$1$ network (that is also strictly tree-based, but not strictfully tree-based). 

Consider Figure \ref{KFully}, which illustrates examples for $k=2,3,4$. Every white vertex in the Figure indicates the vertex has an omitted pendant edge, and every black vertex does not. Then we can see that every cut-vertex without a pendant edge has at least $3$ cut-components (in particular note that the central vertex in (i) has an omitted pendant edge), every other vertex is of degree $3$ and no vertex of degree $4$ or more is adjacent to another one. Constructing level-$k$ fully tree-based graphs is a simple matter of adding additional `diamond formations' around the central cut-vertices.

\begin{figure}[H]
\centering
\begin{tikzpicture}[scale=0.85]
\draw (0,1)--(1,1.5)--(2,1)--(1,0.5)--(0,1);
\draw (2,1)--(3,1.5)--(4,1)--(3,0.5)--(2,1);

\draw[draw=black, fill=white] (2,1) circle [radius=1.5pt];
\draw[draw=black, fill=white] (0,1) circle [radius=1.5pt];
\draw[draw=black, fill=white] (1,1.5) circle [radius=1.5pt];
\draw[draw=black, fill=white] (1,0.5) circle [radius=1.5pt];
\draw[draw=black, fill=white] (4,1) circle [radius=1.5pt];
\draw[draw=black, fill=white] (3,1.5) circle [radius=1.5pt];
\draw[draw=black, fill=white] (3,0.5) circle [radius=1.5pt];
\node at (2,-2) {$\mathstrut (i)$};
\end{tikzpicture}
\hspace{0.5cm}
\begin{tikzpicture}[scale=0.85]
\draw (0,0)--(1,1)--(2,1)--(1,0)--(0,0);
\draw (4,0)--(3,1)--(2,1)--(3,0)--(4,0);
\draw (2,1)--(1.5,2)--(2,3)--(2.5,2)--(2,1);

\draw[fill] (2,1) circle [radius=1.5pt];
\draw[draw=black, fill=white] (0,0) circle [radius=1.5pt];
\draw[draw=black, fill=white] (1,1) circle [radius=1.5pt];
\draw[draw=black, fill=white] (1,0) circle [radius=1.5pt];
\draw[draw=black, fill=white] (4,0) circle [radius=1.5pt];
\draw[draw=black, fill=white] (3,1) circle [radius=1.5pt];
\draw[draw=black, fill=white] (3,0) circle [radius=1.5pt];
\draw[draw=black, fill=white] (1.5,2) circle [radius=1.5pt];
\draw[draw=black, fill=white] (2,3) circle [radius=1.5pt];
\draw[draw=black, fill=white] (2.5,2) circle [radius=1.5pt];

\node at (2,-2) {$\mathstrut (ii)$};
\end{tikzpicture}
\hspace{0.5cm}
\begin{tikzpicture}[scale=0.85]
\draw (2,3)--(1.5,2)--(2,1)--(2.5,2)--(2,3);
\draw (0,1)--(1,1.5)--(2,1)--(1,0.5)--(0,1);
\draw (2,1)--(3,1.5)--(4,1)--(3,0.5)--(2,1);
\draw (2,1)--(1.5,0)--(2,-1)--(2.5,0)--(2,1);

\draw[fill] (2,1) circle [radius=1.5pt];
\draw[draw=black, fill=white] (0,1) circle [radius=1.5pt];
\draw[draw=black, fill=white] (1,1.5) circle [radius=1.5pt];
\draw[draw=black, fill=white] (1,0.5) circle [radius=1.5pt];
\draw[draw=black, fill=white] (4,1) circle [radius=1.5pt];
\draw[draw=black, fill=white] (3,1.5) circle [radius=1.5pt];
\draw[draw=black, fill=white] (3,0.5) circle [radius=1.5pt];
\draw[draw=black, fill=white] (1.5,2) circle [radius=1.5pt];
\draw[draw=black, fill=white] (2,3) circle [radius=1.5pt];
\draw[draw=black, fill=white] (2.5,2) circle [radius=1.5pt];
\draw[draw=black, fill=white] (1.5,0) circle [radius=1.5pt];
\draw[draw=black, fill=white] (2,-1) circle [radius=1.5pt];
\draw[draw=black, fill=white] (2.5,0) circle [radius=1.5pt];

\node at (2,-2) {$\mathstrut (iii)$};
\end{tikzpicture}

\caption{Level-$2,3,$ and $4$ fully tree-based networks. In these diagrams every white vertex indicates the vertex has an omitted pendant edge, and every black vertex does not.}
\label{KFully}
\end{figure}
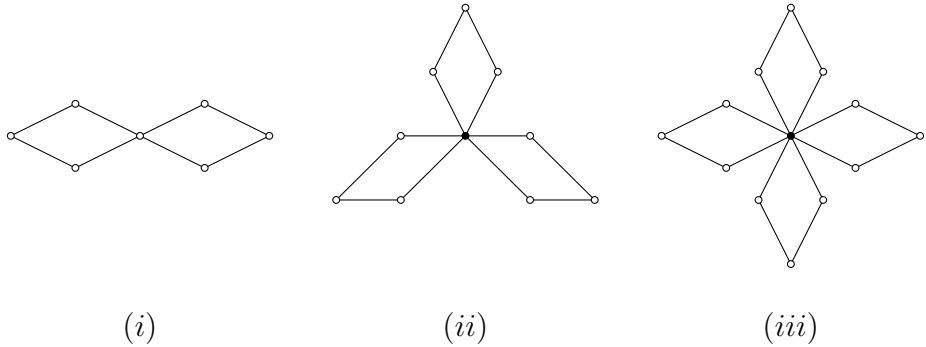

It is also worth noting that, for example, loosefully tree-based networks can also be strictly tree-based without being strictfully tree-based, as Example \ref{CrabEx} illustrates. This is especially pertinent in light of the fact that Figure \ref{TBExamples} (iii) is a strict, fully tree-based network but not a strictfully tree-based network. We further note that Figure \ref{TBExamples} (ii) is an example of a loosefully tree-based network that is tree-based but not strictly tree-based.

\begin{example}
\label{CrabEx}
Figure \ref{CrabGraph} shows an example of a level-$2$ strictly tree-based network $N$ that is loosefully tree-based. To see this, observe that every vertex has at least one pendant edge, so $N$ is loosefully tree-based by Theorem \ref{LoosefullyTreeBased}. A strict base tree may be obtained by deleting edges $A$ and $B$. We note that this example is not fully tree-based, as there exist spanning trees that can be obtained by deleting any two non-pendant edges incident to a single degree $5$ vertex.
\begin{figure}[H]
\centering
\begin{tikzpicture}
\draw (0,1) --(1,2)--(2,3)--(3,3)--(4,2)--(5,3);
\draw (1,2) --(2,1)--(3,1)--(4,2);
\draw (1,2)--(4,2);
\draw (2,1)--(2,0);
\draw (3,1)--(3,0);
\draw (2,3)--(2,4);
\draw (3,3)--(3,4);
\draw (0,3) --(1,2);
\draw (4,2) --(5,1);

\draw[fill] (1,2) circle [radius=1.5pt];
\draw[fill] (2,1) circle [radius=1.5pt];
\draw[fill] (3,3) circle [radius=1.5pt];
\draw[fill] (3,4) circle [radius=1.5pt];
\draw[fill] (3,1) circle [radius=1.5pt];
\draw[fill] (3,0) circle [radius=1.5pt];
\draw[fill] (0,1) circle [radius=1.5pt];
\draw[fill] (0,3) circle [radius=1.5pt];
\draw[fill] (4,2) circle [radius=1.5pt];
\draw[fill] (2,0) circle [radius=1.5pt];
\draw[fill] (2,3) circle [radius=1.5pt];
\draw[fill] (2,4) circle [radius=1.5pt];
\draw[fill] (5,3) circle [radius=1.5pt];
\draw[fill] (5,1) circle [radius=1.5pt];

\node[above=-0.1cm] at (2.5,3) {$\mathstrut A$};
\node[below=0cm] at (2.5,1) {$\mathstrut B$};
\end{tikzpicture}
\caption{A level-$2$ strictly tree-based network that is loosefully tree-based.}
\label{CrabGraph}
\end{figure}
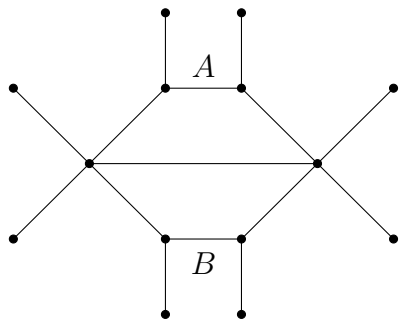
\end{example}

\section{Tree-Based Networks and Colourability}

Let $\chi(G)$ denote the chromatic number of a graph $G$, that is, the minimum number of colours required to colour the vertices of $G$ so that for each edge $(v,w)$, the vertices $v$ and $w$ are coloured differently. In the binary unrooted case, all phylogenetic networks are easily shown to be $3$-colourable as a consequence of Brooks' Theorem, which states that all graphs with maximum degree $d$ are $d$-colourable unless they are complete or an odd cycle \cite{Brooks1941}. However, as we are now examining graphs with no bound on the degree, we require more delicate reasoning to prove results on colourability for tree-based networks.

\begin{theorem}
\label{TBChromatic}
Let $N$ be a network. If $N$ is tree-based, then $N$ is $4$-colourable. However, there exist loosely tree-based networks with chromatic number at least $k$ for any $k > 0$.  
\end{theorem}

\begin{proof}
Suppose $N$ is tree-based. Then $N$ is obtained by taking a tree $T$, adding attachment points to form $T^+$, then adding edges between pairs of attachments points or between an attachment point and a vertex that was in $T$. Now, recall that a tree has chromatic number $2$. Colour the vertices of $T^+$ according to any valid $2$-colouring, and suppose we have $4$ colours available. Then insert edges between the attachment points to obtain $N$ (which may not result in a valid colouring). We then consider the colouring for the attachment points. Let the (only) new edge incident to some attachment point $v_1$ be $e$. There are two possibilities - either $e$ was added between $v_1$ and a vertex of the base tree, or $e$ was added between $v_1$ and another attachment point. 

If $e$ was attached between $v_1$ and a vertex of the base tree, then $v_1$ is adjacent to $3$ other vertices, with up to three different colours. As there are four options for colours, there exists one that we can colour $v_1$ and any conflict with $v_1$ disappears.

Suppose $e$ was attached between $v_1$ and another attachment point $v_2$. As $v_1$ is adjacent to at most three different colours, we can pick the fourth colour and any conflict between $v_1$ and another vertex disappears. Similarly, $v_2$ is only adjacent to three vertices, so we can do the same thing (and may have to - note that $v_2$ may be adjacent to another attachment point we have already recoloured). By following this strategy for every attachment point we see that $N$ is $4$-colourable.

We now construct a loosely tree-based network with chromatic number at least $k$ for some $k \ge 0$. The $k=1,2$ cases are trivially true, as we can just take a single vertex and a tree respectively. Otherwise, take any tree $T$ and insert $k$ attachment points. We can then insert an edge between each attachment point and each other attachment point to form a $k$-clique. It follows that the resulting loosely tree-based network $N$ has $\chi(N) \ge k$.
\end{proof}

We can improve this bound for strictly tree-based networks, but require a technical lemma first.

\begin{lemma}
\label{StrictTechLemma}
Let $N$ be a simple loosely tree-based network that is not a tree, and let $P=\{v_1,...,v_k\}$ be the set of vertices with a pendant edge. Then there is no base tree that contains every non-pendant edge incident to a vertex in $P$. Moreover, if $N$ is strictly tree-based there exists a degree $3$ vertex $v$ with one pendant edge and one incident edge that is not in the base tree.
\end{lemma}

\begin{proof}
Suppose $N$ is a simple loosely tree-based network that is not a tree. Observe that every non-leaf vertex in $N$ must have at least $2$ non-pendant edges, as $N$ is simple. 

Let $N-S$ be a network obtained by deleting some set $S$ of leaves (and their incident edges) so that there is exactly one pendant edge incident to each vertex in $P$ and thus precisely $k$ leaves. Observe that $N-S$ is loosely tree-based with base tree $T-S$ if and only if $N$ is loosely tree-based with base tree $T$. 

If $T-S$ contains every non-pendant edge incident to the vertices in $P$, every vertex in $P$ has degree at least $3$ in $T$, as it must have at least $2$ non-pendant edges and $1$ pendant edge. If we consider just those edges in $T$ incident to vertices in $P$, there are $k$ pendant edges. There are $2k$ non-pendant edges, some of which may be double-counted. As each one can be counted no more than twice, there are at least $\frac{2k}{2}=k$ non-pendant edges incident to the vertices in $P$. Summing the pendant and non-pendant edges, we have $2k$ total edges incident to vertices in $P$. However, $T-S$ is a phylogenetic tree on $k$ leaves, and thus can contain at most $2k-3$ edges, which is a contradiction. It follows that $T-S$ is not a base tree of $N-S$, so $T$ is not a base tree of $N$. Hence $T$ cannot contain all non-pendant edges incident to all vertices in $P$. 

Now, suppose $N$ is strictly tree-based with strict base tree $T$. Then at least one vertex in $P$ has smaller degree in $T$ than it does in $N$, by the first half of the Lemma. However, we know that $T$ must contain all edges incident to vertices of degree $4$ or more (by Definition \ref{SpanDef}), which means that $P$ contains a degree $3$ vertex, with degree $2$ in the base tree. The lemma follows.
\end{proof}

We now demonstrate that strictly tree-based networks are $3$-colourable. This proof contains an induction, and Lemma \ref{StrictTechLemma} is critical for the inductive step.

\begin{theorem}
If $N$ is a strictly tree-based network, then $N$ is $3$-colourable. Moreover there exist strictly tree-based networks of chromatic number $3$.
\end{theorem}

\begin{proof}
We first find a strictly tree-based network of chromatic number $3$. Take the star tree with $3$ leaves, then place attachment points on two of the edges and add an edge between them. The resulting network has chromatic number $3$, as it contains a $3$-clique and it is trivial to find a $3$-colouring.

Suppose $N$ is a strictly tree-based network. We observe that if a graph $G$ is $3$-colourable, the graph $G'$ obtained from $G$ by subdividing an edge $e=(v,w)$ to form $e_1=(v,x), e_2=(x,w)$ is also $3$-colourable. This can be seen by applying the $3$-colouring of $G$ to the corresponding vertices of $G'$, and then observing that $v$ and $w$ are two different colours. We can then set $x$ to be the third colour.

Secondly, observe that it suffices to consider simple networks, as $N$ will be $3$-colourable if and only if every blob $B$ of $N$ is $3$-colourable.

We now proceed by induction on the level of $N$. Suppose $N$ is level-$1$. By definition, this means that we can take a base tree $T$ of $N$, then subdivide two edges of $T$ to form $T^+$, then add an edge between the two attachment points to form $N$. Thus, suppose we have formed $T^+$. As $T^+$ is a tree, it is $2$-colourable, so set a valid $2$-colouring for $T^+$. We then add the edge between attachment points $a$ and $b$. If $a$ and $b$ are different colours, $T$ still has a valid $2$-colouring. Otherwise, if $a$ and $b$ are the same colour, we can just set one of them to the third colour to obtain a valid $3$-colouring.

Now suppose that $N$ is a simple level-$k$ network for $k \ge 2$ and that all level-$(k-1)$ strictly tree-based networks are $3$-colourable. Then select some degree $3$ vertex $v$ together with a strict base tree $T$ so that $v$ has a pendant edge $p$ and edge $e=(v,w)$ in $N$ that is not in $T$, which we can do by Lemma \ref{StrictTechLemma}. We then consider the network $N^-$ obtained by deleting $e$ and suppressing $v$ and $w$ (which are necessarily degree $3$ vertices as $N$ is strictly tree-based). 

We can see that $N^-$ is a simple level-$(k-1)$ network and is thus $3$-colourable by the inductive assumption. If we subdivide the appropriate edges needed to obtain $v$ and $w$ again, then the graph is still $3$-colourable by the observation early in this Theorem. It follows that once we add $e$ back in, the only conflict is potentially between the colouring of $v$ and $w$. In particular, $v$ is adjacent to three vertices, which may be three different colours. However, one vertex adjacent to $v$ is a leaf, so we can set the leaf-vertex to be the same colour as $w$ without generating any more conflicts. Now $v$ is only adjacent to up to two different colours, so we can set $v$ to be the third colour and $N$ is $3$-colourable. It follows that all level-$k$ strictly tree-based networks are $3$-colourable, so by induction, all strictly tree-based networks are $3$-colourable.
\end{proof}

Using this, if we know some network has a subgraph $H$ such that its chromatic number $\chi(H) > 3$ we can immediately say that it is not strictly tree-based. Furthermore, if $\chi(H) >4$, it is not even tree-based.

\section{Final Comments and Further Questions}

In the initial part of the article we extended the current forms of tree-basedness to the unrooted nonbinary setting and defined a new form of tree-basedness, inspired by the spanning tree definition from Francis, Huber and Moulton in \cite{Francis2017}. These were then characterised in terms of spanning trees with particular properties. 

In the second section of the paper, we extended the concept of fully tree-based networks to the unrooted nonbinary setting, characterising each of three possible interpretations.

In the final section we proved some results on colourability of unrooted nonbinary networks.

As unrooted nonbinary tree-based networks have not yet been heavily studied, there are a number of interesting avenues for further research. For instance, given the wide variety of characterisations for tree-based networks in the binary and nonbinary rooted settings \cite{Francis2016,Francis2015,Jetten2016,Zhang2016,Pons2017}, the following natural question arises:

\begin{question} 
Are there analogous characterisations for tree-based networks for the unrooted nonbinary case, especially computationally efficient ones?
\end{question}

Certainly some of these results cannot apply directly, as many rely on the antichain-to-leaf property or modifications thereof, which only makes sense in a rooted setting.

Additionally, several results from \cite{Francis2017} may be worth considering in the new setting.  A \textit{proper} network is a network $N$ for which every cut-edge splits the leaves of $N$. The nonbinary setting also allows for the possibility of cut-vertices that do not have pendant edges, so a suitable extension of the definition of proper would include the requirement that all cut-vertices with $k$ cut-components partition the leaf set into $k$ non-empty subsets. There exist proper level-$1$ networks (in this sense) that are not tree-based and thus not strictly tree-based.  However, the author has yet to find an example of a proper loosely tree-based network (in this sense) of level less than $5$ -  one of level-$5$ is depicted in Figure \ref{NotLoose}. In the binary unrooted setting there are no networks that are not tree-based of level less than $5$ \cite{Francis2017}. 

\begin{question}
Do there exist networks of level less than $5$ that are not loosely tree-based?
\end{question}

Finally, we showed in Theorem \ref{TBChromatic} that if $N$ is a tree-based network, $\chi(N) \le 4$, and in the subsequent theorem that strictly tree-based networks are $3$-colourable. So far the author has yet to find an example of a tree-based network with chromatic number $4$.

\begin{question}
Are tree-based networks $3$-colourable?
\end{question}

Finally, we note that as determining whether an unrooted \textit{binary} network is tree-based is NP-complete, the problem is also NP-complete for determining whether an unrooted nonbinary network has a base tree of any sort, as the definitions coincide in the binary case. It is also not difficult to produce strictly nonbinary networks for which finding a base tree in the loose sense is equivalent to finding a Hamiltonian path, so the strictly nonbinary case is also NP-complete.

\section*{Acknowledgements}
The author would like to thank Professor Andrew Francis of Western Sydney University for his valuable commentary at various stages in development of this paper. The author would also like to acknowledge Western Sydney University for their support in the form of the Australian Postgraduate Award.

\bibliographystyle{amsplain}

\providecommand{\bysame}{\leavevmode\hbox to3em{\hrulefill}\thinspace}
\providecommand{\MR}{\relax\ifhmode\unskip\space\fi MR }
\providecommand{\MRhref}[2]{%
  \href{http://www.ams.org/mathscinet-getitem?mr=#1}{#2}
}
\providecommand{\href}[2]{#2}

\end{document}